\documentclass[a4paper,10pt]{amsart}
\usepackage{amsmath,amssymb,amsfonts}
\usepackage{epsfig}
\usepackage{mathrsfs}
\newtheorem{defn}{Definition}[section]
\newtheorem{thm}{Theorem}[section]
\newtheorem{prop}{Proposition}[section]
\newtheorem{cor}{Corollary}[section]

\newtheorem{rmk}{Remark}[section]
\newtheorem{lma}{Lemma}[section]

\newtheorem{exm}{Example}[section]
\def\N{{\rm I\kern-0.16em N}}
\def\R{{\rm I\kern-0.16em R}}
\def\E{{\rm I\kern-0.16em E}}
\def\P{{\rm I\kern-0.16em P}}
\def\F{{\rm I\kern-0.16em F}}
\def\B{{\rm I\kern-0.16em B}}
\def\C{{\rm I\kern-0.46em C}}
\def\G{{\rm I\kern-0.50em G}}
\newcommand{\ud}{\mathrm{d}}
\numberwithin{equation}{section}
\font\eka=cmex10
\usepackage{ae}
\def\ind{\mathrel{\hbox{\rlap{%
\hbox to 7.5pt{\hrulefill}}\raise6.6pt\hbox{\eka\char'167}}}}
\begin{document}
\title[Option prices with call prices]
{Option prices with call prices}
\author[Viitasaari]{Lauri Viitasaari}
\address{Department of Mathematics and System Analysis, Aalto University School of Science, Helsinki\\
P.O. Box 11100, FIN-00076 Aalto,  FINLAND} 

\begin{abstract}
There exist several methods how more general options can be priced with call prices. In this article, we extend these results to cover a wider class of options and market models. In particular, we introduce a new pricing formula which can be used to price more general options if prices for call
options and digital options are known for every strike price. Moreover, we derive similar results for barrier type options. As a consequence, we obtain a static hedging for general options in the general class of models.
Our result can be utilised in several significant applications. As a simple example, we derive an upper bound for the value of a general American 
option with convex payoff 
and characterise conditions under
which the value of this option equals to the value of the corresponding European
option.

\medskip


\smallskip

\end{abstract}

\maketitle

\section{Introduction}
A major part of the research in valuating option is concentrating on considering options with
convex payoff function. Especially, the call option is widely studied in the literature. 
In this article, we show that in order to valuate and analyse a wide class
of options,
it is sufficient to valuate and analyse only the call and digital option on a certain market model. As our first main result, we show that if the payoff function is continuously differentiable except on at most countable set of points $(s_k)_{k=0}^N$ in which
the left- and the right
limits exists and are finite,
the discounted value $V_t^f$ at time $t$ of a European option with payoff
function $f(X_T)$ is given by
\begin{equation}
 \label{v_general}
\begin{split}
 V_t^f &= B_T^{-1}\E_Q[f(X_T)|\mathcal{F}_t]\\
 &= B_T^{-1}f(0) - \int_0^{\infty} f'(a)\lambda_t(\ud
a)\\
&+ B_T^{-1}\sum_{k=0}^{N}\Delta_-f(s_k)Q(X_T\geq s_k|\mathcal{F}_t)\\
&+ B_T^{-1}\sum_{k=0}^{N}\Delta_+f(s_k)Q(X_T> s_k|\mathcal{F}_t),
\end{split}
\end{equation}
where $B_T$ denotes the bond function, $Q$ denotes the pricing measure, $\lambda_t(a)$ is the discounted value of the call option $(X_T-a)^+$, $\Delta_-f(s_k) = f(s_k)-f(s_k-)$ is the jump on the left and $\Delta_+f(s_k) = f(s_k+)-f(s_k)$ the jump on the right, respectively. As our second main result we provide a similar pricing formula for barrier type options $f(X_T)\textbf{1}_{Y\in C}$, where $Y$ is a random variable (e.g. $Y$ can represent the supremum of the underlying process on $[0,T]$) and $C$ is a Borel set. 

Another problem besides valuation of options is to find a hedging strategy for certain option. One approach is to find a dynamical trading strategy that replicates the payoff. However, then hedging weights have to be adjusted repeatedly which causes problems especially under transaction costs. Another approach is static hedging in which the payoff is replicated with other options such that no re-adjustment is needed. As a consequence of our results, we obtain static hedge for more general options if digital and call options are traded in the market for every strike price. Moreover, our static hedge is applicable for a wide class of models.

The relation between call options and general options is already studied in the literature (for a survey, see section 2). However, many of the existing studies consider cases where either the state space is finite or
the time period is discrete. Our results cover a wider class of options and models than the results of the mentioned
articles. In particular, the only assumption we need is that the pricing is done by taking expectation with respect to some measure $Q$ which is equivalent to the fact that the model is, to some extent, free of arbitrage. For details in the mathematics of arbitrage, we refer to \cite{delbaen} and \cite{schachermayer}.

As an application, we give analogous result for convex payoff functions and derive relatively simple but, to the best of our knowledge, new results on the values of American options. More precisely, we answer to the following
questions:
\begin{enumerate}
 \item
under what conditions on the payoff $f$ and the underlying asset $X_t$, the price of an American
option $f(X_t)$ with
maturity $T$ equals to the price of a European option $f(X_t)$ with maturity $T$?
\item
if the prices are not equal, what is the difference between the prices and which factors
have
influence on the difference?
\end{enumerate}
It turns out that the answer to both questions are determined by the behaviour of the convex function $f$ at the origin whenever
the discounted asset process is a submartingale. In this article, we derive an analytical upper bound 
\begin{equation}
\label{v_upper_a_pos_intro}
\begin{split}
V_t^{f,A} &\leq V_t^f + f(0)_+(B_t^{-1}-B_T^{-1}) \\
&+ f'_+(0)_-
(\overline{X}_t -
\E_Q[\overline{X}_T|\mathcal{F}_t])
\end{split}
\end{equation}
where $A_+ = max(0,A)$, $B_-=min(0,B)$, $\overline{X}_t$ is the discounted underlying
asset, and $V_t^{f,A}$ and
$V_t^f$ denote the discounted values of American and European options.

The rest of the paper is organized as follows. In section 2, we give a short survey of previous studies. In section 3, we introduce the notation and state our main results. The applications are considered in section 4, where we derive the upper bound for the value of American options. Section 5 contains proofs and auxiliary lemmas that we need for the proofs.

\section{Survey of Previous Results}
The first study on the relation between call options and general options is by Breeden and Litzenberger
\cite{Breeden} who showed that the second derivative of a price of European call with
respect to its strike is the
pricing density for more general options provided that the second derivative
exists. We give their result in our notation:

If the second derivative $\lambda_0''(a)$ exists, then the price of European option $f(X_T)$ is given by
\begin{equation}
\label{BL}
V_0^f = \int_{0}^\infty f(a)\lambda_0''(a)\ud a.
\end{equation}

Bick
\cite{Bick} extended this result to a case where either the payoff function or the price
of a call has continuous
second derivative with respect to its strike price except in a finite set of points $(s_k)_{k=0}^N$ in which
the left- and right
derivatives exists and are finite. In particular, Bick showed that
\begin{equation}
\label{kaava_bick}
\begin{split}
 V_0^f &= B_T^{-1}f(0) + \int_0^{\infty} f''(a)\lambda_0(a)\ud a\\
&+ B_T^{-1}\sum_{k=0}^{N}\Delta_-f(s_k)Q(X_T\geq s_k)\\
&+ B_T^{-1}\sum_{k=0}^{N}\Delta_+f(s_k)Q(X_T> s_k)\\
&+ \sum_{k=0}^N (f'(s_k+)-f'(s_k-))\lambda_0(s_k).
\end{split}
\end{equation}

For studies on the relation between call options and general options, also see Jarrow \cite{Jarrow}, who derived a characterisation theorem for the distribution function of the underlying asset, and Brown and Ross \cite{Brown}, who consider a model with finite 
state space and showed that a wide class of
options are a portfolio of call options with different strike prices. 
Similarly, Cox and Rubinstein \cite{Cox} introduced a
method to approximate continuous
functions with piecewise
linear functions, which are a portfolio of call options with different
strikes. They also considered the
pricing error of this approximation, and suggested that one should find an approximation
which is the best in the
sense of maximum absolute difference. However, this may cause problems when considering
infinite state space. 

Static hedging problem has increased its popularity as a subject of research and there exists numerous studies (which we won't list here) with different assumptions. For instance, see Carr and Picron \cite{Carr} and references therein. 

Analytical upper bounds for prices for American options have been developed first by Chen and Yeh \cite{Chen},
then extended by Chang and
Chung \cite{Chang}. However, they do not discuss the properties of the payoff function
itself but instead give
conditions on the value process.

\section{Results}
Let $S_t$ denote the stock price process (or $S_t^k$ in case of several risky assets)
and $X_t$ the underlying asset of an option. As examples, $X_t$ can be
a functional of $S_t$ like the average $X_t = \frac{1}{t}\int_0^t S_u\ud u$
representing Asian option, $X_t = \sup_{u\leq t}S_u$ representing Lookback option, $X_t =
\max_{1\leq k\leq d}S_t^k$
 representing Rainbow option or $X_t = \sum_{k=1}^d \alpha_kS_t^k$ representing Basket
option. Throughout the article,
 $\overline{X}_t$ denotes the discounted value of $X_t$ i.e.
$
\overline{X}_t = B_t^{-1}X_t,
$
where the bond is given by an non-decreasing function $B_t$ with $B_0=1$. The expectation with respect to risk-neutral measure $Q$ is denoted by $\E_Q$. In general, we also
use the terms 'positive' and 'increasing' as synonyms for 'non-negative' and
'non-decreasing'. Similarly, the symbol $\R_+$ refers
to positive real numbers including zero. We also use term the value at time $t$ of an option $f(X)$ with maturity $T$ as synonym for the discounted value at time $t$, and this value is denoted by $V_t^f$ (or $V_t^{f,A}$ for American options, respectively). We omit the dependence on the measure $Q$ on the notation despite the fact that the measure is not necessarily unique. We also omit 
$t$ on the notation whenever we consider the price of the option, i.e. the value at time $t=0$.

We assume that on a certain market model, we are given the underlying asset $X_t$ and the equivalent martingale measure $Q$. We consider the following class of payoff functions.
\begin{defn}
For a function $f:\R_+\rightarrow\R$, we denote $f\in\Pi_Q(X_T)$ if the following conditions are satisfied:
\begin{enumerate}
\item
$f$ is continuously differentiable except on at most countable set of points
$0<s_1<s_2<\ldots<s_N$ (and possibly on $s_0=0$) in which $f$ and $f'$ have jump-discontinuities,
\item
$f(X_T)\in L^1(Q)$,
\item
$f$ satisfies
\begin{equation}
\label{cond1}
\lim_{x\rightarrow\infty}|f(x-)|Q(X_T \geq x) = 0
\end{equation}
and,
\item
the stochastic Riemann-Stieltjes integral
$$
\int_0^{\infty}f'(a)\lambda_t(\ud a)
$$
exists for every $t\in[0,T]$, where
\begin{equation}
\label{mitta_call2}
\lambda_t(a) = \E_Q[B_T^{-1}(X_T-a)^+|\mathcal{F}_t]
\end{equation}
denotes the value of the European call option at $t$ with strike $a$.
\end{enumerate}
\end{defn}
The technical assumptions are not very restrictive.
Indeed, the function $\lambda_t(a)$ is absolutely continuous
and $f'$ has only jump discontinuities. Thus the Riemann-Stieltjes
integral is well-defined on every interval $[0,N]$. Moreover, it coincides with the Lebesgue integral, and the derivative of $\lambda_t(a)$ is the value of a negative digital option with payoff $-\textbf{1}_{x\geq a}$. Note also that 
condition (\ref{cond1}) implies that
$\lim_{x\rightarrow\infty}|f(x-)|Q(X_T \geq x|\mathcal{F}_t) = 0$ almost surely for every $t\in[0,T]$.

For Barrier type options, we consider the class $\Pi_Q(X_T,Y,C)$ and denote $f\in \Pi_Q(X_T,Y,C)$, if: 
\begin{enumerate}
\item
the assumptions $(1)$ and $(2)$ of the class $\Pi_Q(X_T)$ are satisfied, 
\item 
$$
\lim_{x\rightarrow\infty}|f(x-)|Q(X_T \geq x, Y\in C) = 0
$$
and,
\item
the stochastic Riemann-Stieltjes integral
$$
\int_0^{\infty}f'(a)\lambda_t^{Y,C}(\ud a)
$$
exists for every $t\in[0,T]$, where
\begin{equation}
\label{mitta_call3}
\lambda_t^{Y,C}(a) = \E_Q[B_T^{-1}(X_T-a)^+\textbf{1}_{Y\in C}|\mathcal{F}_t].
\end{equation}
\end{enumerate}
Now we can state our main theorems:
\begin{thm}
\label{theor_v_general}
Let $f\in\Pi_Q(X_T)$. 
Then the value of an option $f(X_T)$ at time $t$ is given by
\begin{equation}
 \label{v_general2}
\begin{split}
 V_t^f &= B_T^{-1}f(0) - \int_0^{\infty} f'(a)\lambda_t(\ud
a)\\
&+ B_T^{-1}\sum_{k=0}^{N}\Delta_-f(s_k)Q(X_T\geq s_k|\mathcal{F}_t)\\
&+ B_T^{-1}\sum_{k=0}^{N}\Delta_+f(s_k)Q(X_T> s_k|\mathcal{F}_t),
\end{split}
\end{equation}
where 
$\Delta_-f(s_k) = f(s_k)-f(s_k-)$ and $\Delta_+f(s_k)=f(s_k+)-f(s_k)$. The jump from the left at zero is defined as $\Delta_-f(0)=0$.
\end{thm}
\begin{exm}
\label{exm_power}
As a non-trivial example, consider a power call option for which the payoff is given by $f(x)=(x^n-K)^+$ for some integer $n$. In this case the value of $f(X_T)$ at $t$ is given by
\begin{equation*}
V_t^f = - \int_{K^{\frac{1}{n}}}^{\infty} na^{n-1}\lambda_t(\ud
a).\\
\end{equation*}
\end{exm}
\begin{rmk}
Since $\lambda_t(a)$ is absolutely continuous with respect to $a$, we have
\begin{equation}
\label{abs_leb}
\int_0^{\infty} f'(a)\lambda_t(\ud a) = -\int_0^{\infty} f'(a)B_T^{-1}Q(X_T>a|\mathcal{F}_t)
\ud a.
\end{equation}
Hence the value process $V_t$ of a portfolio
\begin{equation*}
\begin{split}
& f(0) + \int_0^{\infty} f'(a)I_{X_T>a}\ud a\\
&+ \sum_{k=0}^{N}\Delta_-f(s_k)I_{X_T\geq s_k}\\
&+ \sum_{k=0}^{N}\Delta_+f(s_k)I_{X_T>s_k}
\end{split}
\end{equation*}
equals to the value process $V_t^f$ of $f(X_T)$ almost surely for every $t$. Hence
we obtain a static hedging for $f(X_T)$ if we have access to digital options for every strike price. 
\end{rmk}
\begin{rmk}
Note that if the derivative $f'$ is absolutely continuous on every interval $(s_k,s_{k+1})$, then the second derivative of $f$ exists for almost every $a$ and integration by
parts gives the formula (\ref{kaava_bick}). In this case we find a static hedging strategy by investing to call options too. Similarly, if $Q(X_T\geq a|\mathcal{F}_t)$ is absolutely continuous, integration by parts gives (\ref{BL}).
\end{rmk}
\begin{rmk}
If $f$ is a linear combination of convex functions, the integration by parts yields
\begin{equation}
\label{v_conv_e}
 V_t^f = B_T^{-1}f(0) + f'_+(0)\E_Q[\overline{X}_T|\mathcal{F}_t] +
\int_0^\infty \lambda_t(a)\mu(\ud a),
\end{equation}
where $\mu$ is the measure associated with the second derivative of $f$. This also follows directly from the well-known representation \cite{Protter}
\begin{equation}
\label{rep_conv}
f(x) = f(0)+f'_+(0)x+\int_0^\infty (x-a)^+\mu(\ud a)
\end{equation}
for convex functions $f:\R_+\rightarrow \R$. Also note that the formula (\ref{v_conv_e}) holds for every twice continuously differentiable function $f$ (see \cite{Carr}). 
\end{rmk}
\begin{thm}
\label{theor_barrier}
Let $f\in\Pi_Q(X_T,Y,C)$. Then the value $V_t^{f,Y,C}$ of the barrier option $f(X_T)\textbf{1}_{Y\in C}$ at time $t$ is given by
\begin{equation}
 \label{v_general3}
\begin{split}
 V_t^{f,Y,C} &= B_T^{-1}f(0)Q(Y\in C|\mathcal{F}_t) - \int_0^{\infty} f'(a)\lambda_t^{Y,C}(\ud
a)\\
&+ B_T^{-1}\sum_{k=0}^{N}\Delta_-f(s_k)Q(X_T\geq s_k, Y\in C|\mathcal{F}_t)\\
&+ B_T^{-1}\sum_{k=0}^{N}\Delta_+f(s_k)Q(X_T> s_k, Y\in C|\mathcal{F}_t).
\end{split}
\end{equation}
\end{thm}
\begin{rmk}
Note that we can obtain formulas corresponding to (\ref{kaava_bick}) and (\ref{v_conv_e}) for barrier type options as well.
\end{rmk}
\begin{rmk}
In both of our main theorems, we assumed that the bond is a deterministic function. However, 
we may allow the bond to be an adapted, increasing process with obvious modifications in the theorems.
\end{rmk}

\section{Upper Bound for American Option Values}
\begin{prop}
\label{prop_v_a_pos}
Assume that the payoff function $f$ is convex and assume that the discounted underlying
asset
$\overline{X}_t$ is a submartingale.
Then the value of an American option $f(X_t)$ has an upper bound
\begin{equation}
\label{v_upper_a_pos}
\begin{split}
V_t^{f,A} &\leq V_t^f + f(0)_+(B_t^{-1}-B_T^{-1}) \\
&+ f'_+(0)_-
(\overline{X}_t -
\E_Q[\overline{X}_T|\mathcal{F}_t])
\end{split}
\end{equation}
where $A_+ = max(0,A)$ and $B_-=min(0,B)$.
\end{prop}
Note that the Proposition \ref{prop_v_a_pos} can be generalised to a case where the bond is an
adapted, increasing process $B_t$ with $B_0=1$ instead
of deterministic function. In this case, we have to replace the term $B_T^{-1}$ on the upper
bound by conditional
expectation $\E_Q[B_T^{-1}|\mathcal{F}_t]$.

The upper bound becomes impractical if the ratio
\begin{equation}
\label{ratio}
\frac{f(0)_+(B_t^{-1}-B_T^{-1}) + f'_+(0)_- (\overline{X}_t -
\E_Q[\overline{X}_T|\mathcal{F}_t])}{V_t^f}
\end{equation}
is large. For example, the upper bound for the price for put option is inefficient if the option is out 
of the money. However, as a direct corollary
we obtain a useful result which gives sufficient conditions under which the values of
European and American
options are equal. The result in a case where $X_t$ is the stock itself and $f(0)=0$ is already proved in \cite{Shreve}.
\begin{cor}
\label{cor_v_a_pos}
Assume that $f$ is convex and $\overline{X}_t$ is a submartingale. Then the value of an
American option
$f(X_t)$ at time $t$ equals to the value of its European counterpart at time $t$ if the following two conditions hold:
\begin{enumerate}
 \item
bond $B_t$ is a constant 1 or $f(0) \leq 0$,
\item
discounted underlying asset $\overline{X}_t$ is a martingale or $f'_+(0) \geq 0$.
\end{enumerate}
\end{cor}
Note that if $f$ is not linear and
the process $X_t$ has support on all of $(0,\infty)$ (e.g., geometric Brownian motion),
then the Jensen's inequality
 is strict. This implies that the optimal moment to exercise the option is the maturity $T$. 

We have analogous results for options with concave payoff function and the proofs 
are based on the fact $-\inf(-h)=\sup h$ and the arguments in the proof of Proposition \ref{prop_v_a_pos}. In this case, we obtain the conditions under which the value is the discounted instrictic value $B_t^{-1}f(X_t)$. Moreover, if the support of $X_t$ is the whole positive real line, one should exercise the option immediately. 
\begin{exm}
Consider a power option of example \ref{exm_power}.
By Corollary \ref{cor_v_a_pos}, the values of American and European options are the same.
\end{exm}
\begin{exm}
Consider a market model with $n$ risky assets $S_t^1, \ldots, S_t^n$ where $\overline{S}_t^k$ is a martingale
with respect to $Q$ for every $k$ and
consider American Basket option and American Rainbow option that is based on the
stock with the best performance.
For the first one, $\overline{X}_t = \sum_{k=1}^n \alpha_k\overline{S}_t^k$
which is clearly a martingale. For the second one, $\overline{X}_t = \max_{1\leq k \leq n} \overline{S}_t^k$
which is a submartingale (see \cite{Mania}). Hence we may apply Corollary \ref{cor_v_a_pos} to obtain conditions under which the values of American type Basket or Rainbow options equal to their European counterparts. To the best of our knowledge, this is not mentioned in the literature. 
\end{exm}
\begin{rmk}
Usually the values for American options are determined recursively by considering discrete time net $\{0=t_1< t_2<\ldots< t_k=T\}$. In other words, we have
$V_{t_k}^{f,A}=B_{t_k}^{-1}f(X_{t_k})$ and the value $V_{t_{j}}^{f,A}$ for every $j$ is given by
$$
V_{t_{j}}^{f,A} = B_{t_j}^{-1}\max\left(\E_Q[B_{t_{j+1}}V_{t_{j+1}}^{f,A}|\mathcal{F}_{t_j}],f(X_{t_j})\right).
$$
Hence, if we know the prices for European calls with different maturities and strikes and the measure $\mu$, the recursion is relatively easy to implement by 
using values $f(0)$, $f_+'(0)$ and the measure $\mu$ together with European call prices. 
\end{rmk}

\section{Conclusions}
In this article, we have derived a method to valuate wide class of options in terms of call and digital options with different strikes. In particular, typically the traded options have 
payoff functions which belong to our class. Moreover, we have derived similar formula for barrier type options. The benefit of our results is that they are model independent. Indeed, we only assumed that at least one pricing measure exists.  Our results also cover stochastic interest rate models. One can also utilise our pricing formulas to analyse the sensitivity of the price with respect to model parameters.

As an application, we have established an analytical upper bound for the value of an
American option with convex payoff $f$. Surprisingly, the value of American option is determined by the behaviour of the payoff function $f(x)$ at $x=0$. This
fact provides an easy and fast method to analyse the difference between the prices
of American and European
option. However, the upper bound can turn out to be a poor one.

Besides valuating options, it is also of interest to find the hedging strategy.
As a consequence of our results, we obtained a static hedge if the digital options are traded for all strikes. 
Especially, static hedging is of great importance in models with transaction costs. Another hedging problem is to find a dynamic hedging for option. 
Our method can be applied to this problem too, at least to find a candidate for the hedging strategy if we know the hedging strategies for call and 
digital options. However, in this case the assumptions on the model plays role again and we have to consider in which sense the 
stochastic integrals exists and does the call option approximation converge to the hedging candidate. 
For example, it is straightforward to see that following our method in the Black-Scholes model we find a hedge that is the usual delta hedging strategy.

\section{Proofs}
Proof of Theorem \ref{theor_v_general} is based on following lemmas. For
simplicity, we assume in the proofs of lemmas and main theorems
that the bond $B_t$
is a constant one for every $t$. The results
for non-constant bond $B_t$
follows directly by multiplying with $B_T^{-1}$. We also prove formulas only for $t =
0$.
The general case follows by similar arguments.
\begin{lma}
\label{lemma_v_general}
Let $\alpha\geq 0$, $\alpha<\beta<\infty$ and assume that the payoff function $g$
is of
the form $g(x)=f(x)\textbf{1}_{\alpha\leq x\leq \beta}$, where $f$ is continuous on
$[\alpha,\beta]$ and continuously differentiable on $(\alpha,\beta)$. Then the value of an option
$g(X_T)$ at $t$ is given by
\begin{equation}
\begin{split}
V_t^g &= B_T^{-1}f(\alpha)Q(X_T\geq \alpha|\mathcal{F}_t) -
B_T^{-1}f(\beta)Q(X_T > \beta|\mathcal{F}_t)\\
&- \int_\alpha^\beta f'(a) \lambda_t(\ud a).
\end{split}
\end{equation}
\end{lma}
\begin{proof}
Since $g$ is continuous on $[\alpha,\beta]$, we can approximate it with 
\begin{equation}
\label{appro1}
\begin{split}
g_n(x) &= f(\alpha)\textbf{1}_{x=\alpha} + \sum_{k=1}^n (c_kx + b_k)\textbf{1}_{a_k < x \leq a_{k+1}}\\
&= f(\alpha)\textbf{1}_{x=\alpha}+\sum_{k=1}^n (c_kx + b_k)(\textbf{1}_{a_k < x}-\textbf{1}_{a_{k+1}<x}),
\end{split}
\end{equation}
where $\alpha = a_1 < a_2 < \ldots < a_{n+1} = \beta$ is a partition of the interval
$[\alpha,\beta]$ and
the
coefficients are given by
$$
c_k = \frac{f(a_{k+1}) - f(a_k)}{a_{k+1} - a_k},
$$
and
$$
b_k = f(a_{k+1}) - c_ka_{k+1} = f(a_k) - c_ka_k.
$$
The payoff of a call-option with strike $K$ is given by
$$
p(x,K) = (x-K)^+ = x\textbf{1}_{x>K} - K\textbf{1}_{x>K}.
$$
Simple computations yields
\begin{equation}
\label{appro2}
g_n(x) = f(\alpha)\textbf{1}_{x\geq \alpha} - f(\beta)\textbf{1}_{x>\beta}+ \sum_{k=1}^n
c_k\left[p(x,a_k)-p(x,a_{k+1})\right].
\end{equation}
By taking expectation with respect to the equivalent martingale
measure $Q$,
we obtain 
\begin{equation*}
\begin{split}
V^{g_n} &= f(\alpha)Q(X_T \geq \alpha) - f(\beta)Q(X_T >\beta) \\
&+ \sum_{k=1}^n c_k \left[\lambda(a_k) - \lambda(a_{k+1})\right].
\end{split}
\end{equation*}
Now $g_n$ converges to $g$
pointwise and by mean value theorem,
$$
\sum_{k=1}^n c_k \left[\lambda(a_k)
- \lambda(a_{k+1})\right] \rightarrow
- \int_0^b f'(a)\lambda(\ud a)
$$
as $n$ tends to infinity. Applying dominated convergence theorem completes the proof.
\end{proof}
\begin{rmk}
Note that the result can also be obtained by using integration by parts and (\ref{abs_leb}). However, the proof represented here is easier to understand especially in 
the case of barrier options.
\end{rmk}
\begin{lma}
\label{rema_apu}
Let $\alpha\geq 0$, $\alpha<\beta<\infty$ and assume that the payoff function $g$
is of
the form $g^0(x) = f(x)\textbf{1}_{\alpha< x < \beta}$, where $f$ is continuous on
$[\alpha,\beta]$ and continuously differentiable on $(\alpha,\beta)$. Then the value of an option
$g^0(X_T)$ at $t$ is given by
\begin{equation}
\begin{split}
V_t^{g^0} &= B_T^{-1}f(\alpha+)Q(X_T> \alpha|\mathcal{F}_t) - B_T^{-1}f(\beta-)Q(X_T \geq
\beta|\mathcal{F}_t)\\
&- \int_\alpha^\beta f'(a) \lambda_t(\ud a).
\end{split}
\end{equation}
\end{lma}
\begin{proof}
Define a new function $g$ by
\begin{equation*}
g(x) = \begin{cases}
f(\alpha+),& x = \alpha\\
f(x),& x \in (\alpha,\beta)\\
f(\beta-),& x = \beta.
\end{cases}
\end{equation*}
By Lemma \ref{lemma_v_general}, we 
obtain
\begin{equation*}
\begin{split}
V_t^g &= B_T^{-1}f(\alpha)Q(X_T\geq \alpha|\mathcal{F}_t) - B_T^{-1}f(\beta)Q(X_T >
\beta|\mathcal{F}_t) \\
&- \int_\alpha^\beta f'(a) \lambda_t(\ud a).
\end{split}
\end{equation*}
Noting that $g^0(x) = g(x) - g(\alpha)\textbf{1}_{x=\alpha}-g(\beta)\textbf{1}_{x=\beta}$ we obtain the result.
\end{proof}
\begin{proof}[Proof of Theorem \ref{theor_v_general}]
Put
$g_b(x) = f(x)\textbf{1}_{0\leq x < b}$, where $s_{n+1}=b$.
By assumptions, we may write
$$
g_b(x) = \sum_{k=0}^n f(x)\textbf{1}_{s_k<x<s_{k+1}} + \sum_{k=0}^n f(x)\textbf{1}_{x=s_k}.
$$
For terms on the
first sum we obtain by Lemma \ref{rema_apu}
\begin{equation*}
\begin{split}
 V_t^{g_b} &= B_T^{-1}f(0) - \int_0^b f'(a)\lambda_t(\ud
a)\\
&+ B_T^{-1}\sum_{k=0}^{n}\Delta_-f(s_k)Q(X_T\geq s_k|\mathcal{F}_t)\\
&+ B_T^{-1}\sum_{k=0}^{n}\Delta_+f(s_k)Q(X_T> s_k|\mathcal{F}_t)\\
&- B_T^{-1}f(b-)Q(X_T\geq b|\mathcal{F}_t).
\end{split}
\end{equation*}
Letting $b\rightarrow\infty$, applying dominated convergence theorem and taking account the assumptions $(3)$ and $(4)$ on $\Pi_Q(X_T)$ we obtain the result.
\end{proof}
\begin{proof}[Proof of Theorem \ref{theor_barrier}]
By multiplying (\ref{appro2}) with $\textbf{1}_{Y\in C}$ and following proofs of 
Lemma \ref{lemma_v_general} and Lemma \ref{rema_apu} we obtain similar pricing formulas for options
$g(X_T)\textbf{1}_{Y\in C}$ and $g^0(X_T)\textbf{1}_{Y\in C}$. Hence, by following the proof of Theorem \ref{theor_v_general}, we obtain the result.
\end{proof} 
\begin{proof}[Proof of Proposition \ref{prop_v_a_pos}]
The proof of \ref{prop_v_a_pos} relies on the well-known result for call options and the representation (\ref{rep_conv}).
Therefore we
begin with introducing the result for call options. For the proof, we refer to \cite{Chesney}.
\begin{lma}
\label{lem_call_a}
Assume that the discounted underlying asset $\overline{X}_t$ is a submartingale. Then the
value at time $t$ of an
American call option $(X_t - K)^+$ with terminal time $T$ is the same as the value at
time $t$ of a European call
option $(X_T - K)^+$.
\end{lma}
By Tonelli's Theorem, Lemma \ref{lem_call_a} and representation (\ref{rep_conv}) we obtain
\begin{equation*}
\begin{split}
&\sup_{\tau\in[t,T]}\E_Q[B_{\tau}^{-1}f(X_{\tau})|\mathcal{F}_t] \\
&\leq \sup_{\tau\in[t,T]}\E_Q[B_{\tau}^{-1}f(0)|\mathcal{F}_t] +
\sup_{\tau\in[t,T]}\E_Q[f'_+(0)B_{\tau}^{-1}
X_{\tau}|\mathcal{F}_t] \\
&+ \sup_{\tau\in[t,T]}\int_0^\infty\E_Q[B_{\tau}^{-1}(X_{\tau}-K)^+|\mathcal{F}_t]\mu(\ud
a)\\
&\leq \sup_{\tau\in[t,T]}\E_Q[B_{\tau}^{-1}f(0)|\mathcal{F}_t] +
\sup_{\tau\in[t,T]}\E_Q[f'_+(0)B_{\tau}^{-1}
X_{\tau}|\mathcal{F}_t] \\
&+ \int_0^\infty\sup_{\tau\in[t,T]}\E_Q[B_{\tau}^{-1}(X_{\tau}-K)^+|\mathcal{F}_t]\mu(\ud
a)\\
&= \sup_{\tau\in[t,T]}\E_Q[B_{\tau}^{-1}f(0)|\mathcal{F}_t] +
\sup_{\tau\in[t,T]}\E_Q[f'_+(0)\overline{X}_{\tau}
|\mathcal{F}_t] \\
&+ \int_0^\infty\E_Q[B_T^{-1}(X_T-K)^+|\mathcal{F}_t]\mu(\ud a).
\end{split}
\end{equation*}
Thus, by using (\ref{v_conv_e}), we obtain
\begin{equation*}
\begin{split}
V_t^{f,A} -V_t^f &\leq 
\sup_{\tau\in[t,T]}\E_Q[B_{\tau}^{-1}f(0)|\mathcal{F}_t] - B_T^{-1}f(0)\\
 &+
\sup_{\tau\in[t,T]}\E_Q[f'_+(0)\overline{X}_{\tau}|\mathcal{F}_t] -
\E_Q[f'_+(0)\overline{X}_T|\mathcal{F}_t].
\end{split}
\end{equation*}
Evidently, 
$$
\sup_{\tau\in[t,T]}\E_Q[B_{\tau}^{-1}f(0)|\mathcal{F}_t] - B_T^{-1}f(0) = f(0)_+(B_t^{-1}-B_T^{-1}).
$$
Similarly, by the submartingale property of $\overline{X}_t$, we obtain
$$
\sup_{\tau\in[t,T]}\E_Q[f'_+(0)\overline{X}_{\tau}|\mathcal{F}_t] -
\E_Q[f'_+(0)\overline{X}_T|\mathcal{F}_t] = f'_+(0)_-(\overline{X}_t-\E_Q[\overline{X}_T|\mathcal{F}_t]).
$$
\end{proof}

\vskip0.5cm

\textbf{Acknowledgements}.\\
The author thanks Esko Valkeila and Jussi Keppo for discussions and comments which improved the paper. The author thanks the Finnish Doctoral Programme in Stochastics and Statistics for financial support.





\end{document}